\documentclass[11pt,english]{article}
\usepackage[T1]{fontenc}
\usepackage[latin9]{inputenc}
\usepackage{geometry}
\usepackage{mathrsfs}
\usepackage{algorithm2e}
\usepackage{amsmath}
\usepackage{amsthm}
\usepackage{amssymb}
\usepackage{graphicx}

\makeatletter
\numberwithin{equation}{section}
\numberwithin{figure}{section}
\theoremstyle{plain}
\newtheorem{thm}{\protect\theoremname}[section]
\theoremstyle{definition}
\newtheorem{defn}[thm]{\protect\definitionname}
\theoremstyle{plain}
\newtheorem{lem}[thm]{\protect\lemmaname}
\theoremstyle{plain}
\newtheorem{prop}[thm]{\protect\propositionname}
\theoremstyle{plain}
\newtheorem{cor}[thm]{\protect\corollaryname}

\SetAlgoNoEnd
\RestyleAlgo{boxed}
\LinesNumbered
\usepackage{float}
\usepackage[framemethod=tikz]{mdframed}

\mdfdefinestyle{myFigureBoxStyle}{tikzsetting={draw=black, line width=1pt}}%

\makeatletter
  \newcommand\fs@myRoundBox{\def\@fs@cfont{\bfseries}\let\@fs@capt\floatc@plain
  \def\@fs@pre{\begin{mdframed}[style=myFigureBoxStyle]}%
  \def\@fs@mid{\vspace{\abovecaptionskip}}%
  \def\@fs@post{\end{mdframed}}\let\@fs@iftopcapt\iffalse}
\makeatother

\floatstyle{myRoundBox} 
\restylefloat{figure}
\date{}
\usepackage{enumitem}
\setlist{nosep}
\usepackage{titlesec}
\titlespacing*{\section}
{0pt}{1ex}{1ex}
\titlespacing*{\subsection}
{0pt}{0.2ex}{0.2ex}

\makeatother

\usepackage{babel}
\providecommand{\corollaryname}{Corollary}
\providecommand{\definitionname}{Definition}
\providecommand{\lemmaname}{Lemma}
\providecommand{\propositionname}{Proposition}
\providecommand{\theoremname}{Theorem}

\begin{document}
\global\long\def\l{\ell}

\global\long\def\N{\mathbb{N}}

\global\long\def\Z{\mathbb{Z}}

\global\long\def\R{\mathbb{R}}

\global\long\def\C{\mathbb{C}}

\global\long\def\defi{\triangleq}

\global\long\def\s{\subseteq}

\global\long\def\w{\bar{w}}

\global\long\def\A{\mathcal{A}}

\global\long\def\F{\mathcal{F}}

\global\long\def\J{\mathcal{J}}

\global\long\def\W{\mathscr{W}}

\global\long\def\ttop{\text{top}}

\title{Improved Online Algorithm for Weighted Flow Time}

\author{Yossi Azar\\azar@tau.ac.il\\Tel Aviv University\thanks{Partially supported by the ISF grant no. 1506/16 and by the ICRC Blavatnik Fund.} \and Noam Touitou\\noamtouitou@mail.tau.ac.il\\Tel Aviv University}
\maketitle
\begin{abstract}
We discuss one of the most fundamental scheduling problem of processing
jobs on a single machine to minimize the weighted flow time (weighted
response time). Our main result is a $O(\log P)$-competitive algorithm,
where $P$ is the maximum-to-minimum processing time ratio, improving
upon the $O(\log^{2}P)$-competitive algorithm of Chekuri, Khanna
and Zhu (STOC 2001). We also design a $O(\log D)$-competitive algorithm,
where $D$ is the maximum-to-minimum density ratio of jobs. Finally,
we show how to combine these results with the result of Bansal and
Dhamdhere (SODA 2003) to achieve a $O(\log(\min(P,D,W)))$-competitive
algorithm (where $W$ is the maximum-to-minimum weight ratio), without
knowing $P,D,W$ in advance. As shown by Bansal and Chan (SODA 2009),
no constant-competitive algorithm is achievable for this problem.\vfill
\end{abstract}
\pagebreak{}

\section{Introduction}

We discuss the fundamental problem of online scheduling of jobs on
a single machine. In this problem, a machine receives jobs over time,
and the online algorithm decides which job to process at any point
in time, allowing preemption as needed. Each of the jobs has a processing
time (or volume) and a weight. We consider the weighted flow time
cost function for this problem, in which the algorithm has to minimize
the weighted sum over the jobs of the time between their arrival and
their completion.

Define $P$ as the maximal ratio between the processing times of any
two jobs, and $W$ as the maximal ratio between the weights of any
two jobs. The first non-trivial algorithm for this problem was shown
by Chekuri et al. \cite{DBLP:conf/stoc/ChekuriKZ01}, and is $O(\log^{2}P)$-competitive.
Bansal and Dhamdhere \cite{DBLP:journals/talg/BansalD07} showed a
$O(\log W)$-competitive algorithm for this problem. A lower bound
on competitiveness of \\
$\Omega(\min((\log W/\log\log W)^{\frac{1}{2}},(\log\log P/\log\log\log P)^{\frac{1}{2}}))$
was then shown by Bansal and Chan \cite{DBLP:conf/soda/BansalC09}.

An additional parameter of interest is $D$, which is the maximal
ratio between the densities of any two jobs. This parameter has not
been explored in any previous work, though the algorithm in \cite{DBLP:conf/stoc/ChekuriKZ01}
can be modified quite easily to yield a $O(\log^{2}D)$-competitive
algorithm.

\subsection*{Our Results}

Our main results are as follows:
\begin{itemize}
\item $O(\log P)$-competitive algorithm for weighted flow time on a single
machine, improving upon the previous result of $O(\log^{2}P)$ in
\cite{DBLP:conf/stoc/ChekuriKZ01}. 
\item $O(\log D)$-competitive algorithm. This algorithm is different from
the $O(\log P)$-competitive algorithm. 
\item A combined algorithm which is $O(\log(\min(P,D,W)))$ competitive,
without knowing $P,D,W$ in advance. This builds on the previous two
algorithms and the algorithm of \cite{DBLP:journals/talg/BansalD07}.
\end{itemize}

\subsection*{Related Work}

As mentioned above, for minimizing weighted flow time on a single
machine, Chekuri et al. \cite{DBLP:conf/stoc/ChekuriKZ01} presented
a $O(\log^{2}P)$-competitive algorithm. The algorithm arranges all
jobs in a table by their weight and density. The algorithm only processes
jobs the weight of which is larger than the sum of the weights in
the rectangle of lower density and lower weight jobs. In \cite{DBLP:journals/talg/BansalD07},
a $O(\log W)$-competitive algorithm is presented for this problem.
This algorithm is based on the division of jobs into bins according
to the logarithmic class of their weight, and then balancing the weights
of those bins. A lower bound on algorithm competitiveness of $\Omega(\min((\log W/\log\log W)^{\frac{1}{2}},(\log\log P/\log\log\log P)^{\frac{1}{2}}))$
was shown in \cite{DBLP:conf/soda/BansalC09}. Typically, one assumes
that $P$ and $W$ are bounded and that the number of jobs may be
arbitrarily large (unbounded). In the case that the number of jobs
is small, it is shown in \cite{DBLP:journals/talg/BansalD07} how
to use the $O(\log W)$ competitive algorithm to construct a $O(\log P+\log n)$-competitive
algorithm, with $n$ being the number of jobs released. Note that
the competitive ratio of this algorithm is unbounded as a function
of $P$. For the offline problem, Chekuri and Khanna \cite{Chekuri:2002:ASP:509907.509954}
showed a quasi-polynomial time approximation scheme. Bansal \cite{Bansal:2005:MFT:2308916.2309425}
later extended this result to a constant number of machines.

For a single machine in the unweighted case, a classic result from
\cite{NAV:NAV3800030106} states that the shortest remaining processing
time algorithm is optimal. For the case of $m$ identical machines,
and a set of $n$ jobs, Leonardi and Raz \cite{Leonardi:1997:ATF:258533.258562}
showed that SRPT is $O(\log(\min(P,\frac{n}{m})))$-competitive, and
that this is optimal. Awerbuch et al. \cite{doi:10.1137/S009753970037446X}
presented a $O(\log(\min(P,n)))$-competitive algorithm in which jobs
do not migrate machines. An algorithm with immediate dispatching (and
no migration) with similar guarantees was presented in \cite{Avrahami:2003:MTF:777412.777415}.
In contrast to those positive results, for the case of weighted flow
time on multiple machines, Chekuri et al. \cite{DBLP:conf/stoc/ChekuriKZ01}
showed a $\Omega(\min(\sqrt{P},\sqrt{W},(\frac{n}{m})^{\frac{1}{4}}))$
lower bound on competitiveness for $m>1$ machines, making the problem
intractable.

Competitive algorithms also exist for unweighted flow time on related
machines \cite{DBLP:conf/icalp/GargK06,DBLP:conf/stoc/GargK06}. For
the case of unweighted flow time for machines with restricted assignment,
no online algorithm with a bounded competitive ratio exists \cite{DBLP:conf/focs/GargK07}. 

In the resource augmentation problem, the single machine weighted
flow time problem becomes significally easier. In \cite{BECCHETTI2006339},
a $(1+\frac{1}{\epsilon})$-competitive algorithm is given for a $(1+\epsilon)$-speed
model. An algorithm with a similar guarantee was given in \cite{DBLP:journals/talg/BansalD07}
for the non-clairvoyant setting. A competitive algorithm is also known
for weighted flow time in unrelated machines \cite{DBLP:conf/stoc/ChadhaGKM09}.
Additional work has been done on weighted flow time for unrelated
machines with resource augmentation in the offline model (see e.g.
\cite{DBLP:conf/stoc/ChekuriGKK04,doi:10.1137/1.9781611973099.97}).

\subsection*{Our Technique}

\textbf{Processing-time ratio algorithm.} Our $O(\log P)$ algorithm
rounds the weights of incoming jobs to a power of $2$, and assigns
them into power-of-$2$ bins according to their processing time. That
is, bin $i$ contains jobs whose initial processing time is in the
range $[2^{i},2^{i+1})$. Note that jobs never move between bins. 

To describe the algorithm, we utilize a visualization of the algorithms
state, as shown in figure \ref{fig:Processing-Time-Bins}. We view
each job as a rectangular container, such that the height of the rectangle
is the weight of the job, and its area is $2^{i+1}$. The volume of
the job can be viewed as liquid within that container. When a job
is being processed, the amount of liquid is reduced from the top.
A horizontal dotted line runs through the middle of the container.
This line helps to observe the volume covered by a bar, a concept
used in the algorithm's analysis.

Inside each bin, jobs of higher weight have a higher priority for
processing. For jobs of the same weight, the jobs are ordered according
to their density, such that lower density jobs get priority. This
priority is illustrated by the rectangular containers of the jobs
being stacked on top of one another, such that a higher job in the
stack has higher priority. 

At any point in time, the algorithm chooses a bin from which to process
the uppermost job. At each point in time, each bin is assigned a score,
such that the bin with the highest score is processed. In \cite{DBLP:journals/talg/BansalD07},
the score assigned to each bin is the total weight of the jobs in
that bin. In our algorithm, the score is more complex. All jobs except
the top job add their weight to the total score of the bin. The top
job adds to the score of the bin either its complete weight, if it
has high remaining processing time, or half of it, if it has low remaining
processing time. While this modification seems odd, it is crucial
for the proof. This can lead to some interesting behavior: for example,
a job can be preempted during processing because its processing time
has decreased below the threshold, lowering the score of its bin.
This preemption is not due to any external event; no job has been
released to trigger it.

The outline of our proof is inspired by \cite{DBLP:journals/talg/BansalD07}.
In our proof, we use the concept of volume covered by a bar. We place
a horizontal bar at some height $x$, as shown in figure \ref{fig:Processing-Time-Bins}.
If the bar lies over the top of a job, we say that it covers the entire
volume of the job. If the bar lies between the half of the job (the
dotted lines in figure \ref{fig:Processing-Time-Bins}) and the top
of the job, it only covers the volume underneath the half of the job.
Otherwise, the bar covers none of the job's volume.

The proof consists of showing that at any time, a bar at a height
that is constant times the remaining weight of the optimum, every
job in the algorithm covers the entire volume in the algorithm. We
then show that such a bar can be raised by a multiplicative factor
of $2$ to yield a second bar, which covers all weight in all bins
(and therefore, its height is larger than the total weight in that
bin). This gives us the $O(\log P)$ competitiveness.

\textbf{Density ratio algorithm. }A different algorithm, though similar
in the method of its proof, is used for $O(\log D)$ competitiveness.
We assign the jobs into power-of-$2$ bins according to density. In
this algorithm, the weights of the algorithm are not rounded to a
power of $2$. Instead, the densities of the jobs are rounded to a
power of $2$ (this is also through changing the weights). This gives
us that bin $i$ contains jobs whose initial density is $2^{i}$.
In this algorithm, as in the previous algorithm, jobs never move between
bins. 

A visualization of a possible state of the algorithm at a time $t$
is shown in figure \ref{fig:Density-Bins}. We again have a rectangular
container for each job, the height of which is the weight of the job.
In this case, the container arrives full\textendash its width is $2^{i}$,
and thus its area is exactly the volume of the job upon its arrival.
As in the $O(\log P)$ algorithm, when a job is processed its container
depletes from the top. Inside each bin, one job takes priority over
another if its weight is of a higher power of $2$ (higher weight
class). For jobs of the same weight class, a partially processed job
takes priority. 

In this algorithm, each bin has a single dummy job of each weight
class. To calculate the score of a bin, we observe its top (non-dummy)
job. The jobs below that top job\textendash dummy or not\textendash add
their weight to the score. The top job, however, adds to the score
only a fraction of its weight, which is the fraction of its volume
that has not yet been processed. Thus, the score of a bin depends
continuously on the processing state of the top job, as opposed to
the two discrete options in the $O(\log P)$ algorithm.

This score again stems from a definition for the volume covered by
a bar. In this case, the definition is continuous\textendash the volume
covered by the bar is exactly the volume that appears under the bar
in a visualization such as figure \ref{fig:Density-Bins}. 

\textbf{Combined algorithm. }Finally, we use the fact that the $O(\log W)$
algorithm of \cite{DBLP:journals/talg/BansalD07} and the $O(\log P)$
and $O(\log D)$ algorithms have some common properties, and show
how to combine them into a $O(\log(\min(P,D,W)))$ competitive algorithm
without knowledge of $P,D,W$ in advance. This relies on bins of three
types coexisting, and being assigned jobs in a manner that prefers
bins that have already been in use. We note that this method works
by using the specific common properties of the different algorithms,
and cannot be applied to general algorithms for these three parameters.

\section{Preliminaries}

\subsection{The Model}

We are given a single machine, and jobs arrive over time. The machine
can choose which job to process at any given time, with the option
to preempt a previous job if necessary.

An instance in the model is a set of jobs $\mathcal{J}$, so that
every job $J\in\mathcal{J}$ of index $I(J)\in\N$ has the following
attributes:
\begin{itemize}
\item A processing time $p(J)>0$ (also called the volume of $J$).
\item A weight $w(J)>0$, which represents the significance of the job.
\item A time of release $r(J)\ge0$. The job $J$ cannot be processed prior
to this time. 
\end{itemize}
Let the density of a job $J$ be $d(J)=\frac{p(J)}{w(J)}$. We consider
the online clairvoyant model, in which an algorithm is not aware of
the existence of a job $J$ of index $I(J)$ at time $t<r(J)$, but
knows $p(J),w(J)$ once $t\ge r(J)$. 

For a given algorithm and an instance $\mathcal{J}$, denote by $c(J)$
the completion time of a job $J\in\mathcal{J}$ in the algorithm.
The goal of the algorithm is to minimize weighted flow time, defined
as:
\[
\text{\ensuremath{\mathscr{C}}}(\mathcal{J})=\sum_{J\in\mathcal{J}}w(J)\cdot(c(J)-r(J))
\]

Note that when all $w(J)=1$ this reduces to the flow time of the
algorithm, which is the term $\sum_{J\in\J}(c(J)-r(J))$.

For a specific instance $\J$, we define the parameters $P=\max_{J_{1},J_{2}\in\mathcal{J}}\frac{p(J_{1})}{p(J_{2})}$,$W=\max_{J_{1},J_{2}\in\mathcal{J}}\frac{w(J_{1})}{w(J_{2})}$
and $D=\max_{J_{1},J_{2}\in\mathcal{J}}\frac{d(J_{1})}{d(J_{2})}$.
Our algorithms do not require knowledge of these parameters in advance.

\subsection{Common Definitions}

All algorithms we propose consist of two parts. The first part assigns
every job to a bin upon release. The second part decides which job
to process at any given time, based on the contents of those bins.

Denoting by $\A$ the set of bins used by our algorithm, we use the
following definitions.

\begin{defn}
Define the following:
\begin{itemize}
\item For every bin $A\in\mathcal{A}$ and a point in time $t$, define
$A(t)$ as the set of jobs alive in the algorithm in $A$ at $t$
(i.e. jobs that have arrived but have not yet been completed).
\item Denote by $Q(t)=\bigcup_{A\in\A}A(t)$ the set of all jobs alive at
time $t$ in the algorithm. 
\item For a set of jobs $S$, define $w(S)=\sum_{J\in S}w(J)$. Define $W(t)=w(Q(t))$
\item Define $p_{t}(J)$ to be the remaining processing time of $J$ (also
called the remaining volume of $J$) at time $t$.
\item For a set of jobs $S$ alive at time $t$, define $p_{t}(S)=\sum_{J\in S}p_{t}(J).$
Define $V(t)=p_{t}(Q(t))$. 
\end{itemize}
\end{defn}

When considering an optimal algorithm, we refer to its attributes
by adding an asterisk to the aforementioned notation. For example,
$W^{*}(t)$ is the living weight in the optimal algorithm at time
$t$, and $p_{t}^{*}(J)$ is the remaining processing time of job
$J$ in the optimal algorithm  at time $t$.

Our proofs use the following observation to show competitiveness.

\theoremstyle{theorem}
\newtheorem{obs}[thm]{Observation}

\begin{obs}\label{obs:LocalCompImpliesGlobal}

If $W(t)\le c\cdot W^{*}(t)$ for all $t$ then the algorithm is $c$-competitive.
This is due to the fact that the weighted flow time of the algorithm
can be expressed as $\int_{0}^{\infty}W(t)dt$.

 \end{obs}

\section{\label{sec:ProcessingTimeAlgo}The $O(\log P)$-Competitive Algorithm}

In this section, we present a $O(\log P)$-competitive algorithm for
weighted flow time. The algorithm consists of two parts. The first
part, assigns incoming jobs immediately to a bin from the set $\A=\{A_{i}\,|\,i\in\Z\}$.
The second part processes jobs from the bins.

The algorithm \textbf{assumes wlog }that every job $J$ arrives with
a weight that is an integral power of $2$. The algorithm can enforce
this by rounding the weight of every job $J$ to $2^{\lg(w(J))+1}$,
which adds a factor of $2$ to its competitive ratio.

Within each bin, we have an ordering between jobs that determines
their priority. 
\begin{defn}
At a time $t$ and bin $A$:
\begin{itemize}
\item For $J_{1},J_{2}$ in $A(t)$, we write $J_{1}\prec_{t}J_{2}$ when
$w(J_{2})>w(J_{1})$ or when $w(J_{2})=w(J_{1})$ and $p_{t}(J_{1})>p_{t}(J_{2})$.
If $w(J_{2})=w(J_{1})$ and $p_{t}(J_{1})=p_{t}(J_{2})$, we break
ties according to the indices of the jobs, $I(J_{1})$ and $I(J_{2})$
(i.e. arbitrarily).
\item We denote by $\ttop_{A}(t)$ the maximal job with respect to $\prec_{t}$.
\end{itemize}
\end{defn}

\begin{defn}
For a bin $A_{i}\in\A$ and a time $t$:
\begin{itemize}
\item For a job $J$ assigned to $A_{i}$, we say that $J$ is \emph{well-processed}
if $p_{t}(J)\le2^{i}$. Denote by $\delta(J)$ the indicator variable
for being well-processed.
\item Define the \emph{score} of bin $A_{i}$ to be $\W_{A_{i}}(t)=w(A_{i}(t))-\delta(\ttop_{A_{i}}(t))\cdot\frac{w(\ttop_{A_{i}}(t))}{2}$.
\end{itemize}
\end{defn}

Algorithm \ref{alg:PAlgo} as described below is $O(\log P)$ competitive.

\SetAlgoNoEnd

\begin{algorithm}
\caption{\label{alg:PAlgo}$O(\log P)$ Competitive}

Whenever a new job $J$ arrives:\\
\Indp
assign $J$ to $A_{i}$ such that $2^i<p(J)\le 2^{i+1}$.\\
\Indm
At any time $t$:\\
\Indp
For $A=\arg \max _A(\mathscr{W}_A(t))$, process $\text{top}_A(t)$.
\end{algorithm}

\section{\label{sec:AnalysisPT}Analysis of $O(\log P)$-Competitive Algorithm}

Consider a visualization such as figure \ref{fig:Processing-Time-Bins},
of the state of the algorithm at a time $t$. The \emph{base }of a
job $J\in A_{i}(t)$ (denoted $\beta(J,t)$) is the height of its
bottom in the visualization. Formally, $\beta(J,t)=w(\{J^{\prime}\in A_{i}(t)\,|\,J^{\prime}\prec_{t}J\})$. 

Consider a horizontal bar at height $x$. We now define \emph{the
volume covered by $x$}, or \emph{the volume under $x$. }This is:
\begin{itemize}
\item All the volume of jobs completely under $x$.
\item None of the volume of jobs completely over $x$.
\item Some of the volume of jobs that intersect $x$.
\end{itemize}
Formally:
\begin{defn}
At any time $t$, and a bin $A_{i}\in\A$:
\begin{itemize}
\item For any job $J\in A_{i}(t)$ and a non-negative $x$, the \emph{volume
of $J$ under }$x$ is: 
\[
\gamma_{J}(x,t)=\begin{cases}
p_{t}(J) & x\ge\beta(J,t)+w(J)\\
0 & x<\beta(J,t)+\frac{w(J)}{2}\\
\min(p_{t}(J),2^{i}) & \beta(J,t)+\frac{w(J)}{2}\le x<\beta(J,t)+w(J)
\end{cases}
\]
\item The \emph{volume of bin $A_{i}$ under $x$ }is:
\[
B_{A_{i}}(x,t)=\sum_{J\in A_{i}(t)}\gamma_{J}(x,t)
\]
\\
\item The \emph{total volume under $x$ }is\emph{:}
\[
B(x,t)=\sum_{A_{i}\in\A}B_{A_{i}}(x,t)
\]
\end{itemize}
\end{defn}

The previous definition yields the following observation. 

\begin{obs}\label{obs:HeightIsMinimalCover}

$\W_{A_{i}}(t)$ as defined before is the minimal bar that covers
the entire volume of a bin $A_{i}$.

 \end{obs}

In this section, we prove the following theorem.
\begin{thm}
\label{thm:PTCompetitive}The algorithm described in Section \ref{sec:ProcessingTimeAlgo}
is $O(\log P)$-competitive.
\end{thm}

\begin{figure}[!t]
\caption{\label{fig:Processing-Time-Bins}Processing Time Bins}

The following image is a possible state of two bins in the algorithm
at some time $t$.

\includegraphics[scale=0.8]{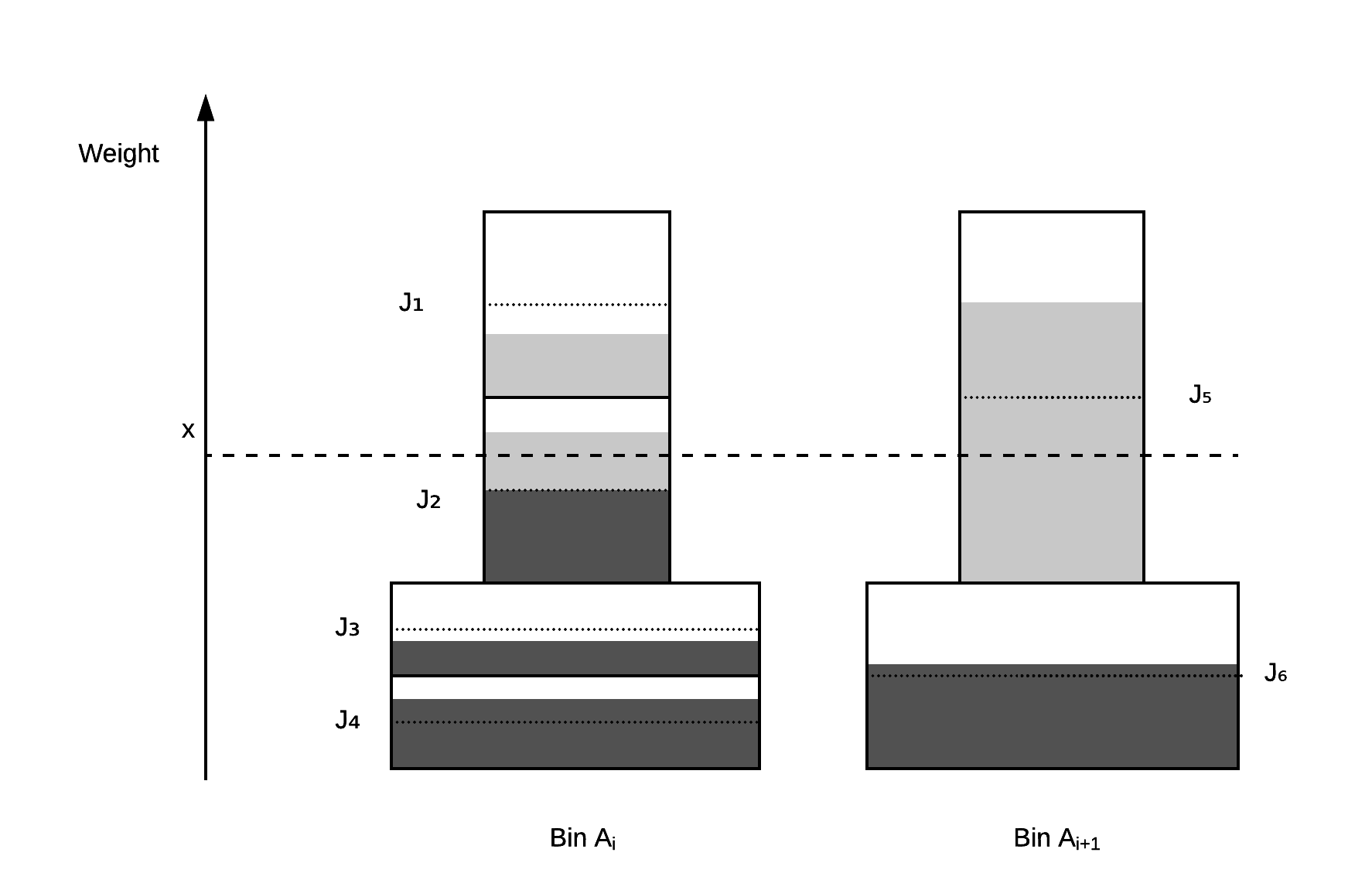}\\

\begin{raggedright}
In this figure, each job is described by a rectangular container.
The height of the container is the weight of the job, and its area
is $2^{j+1}$ for bin $A_{j}$. The area of the gray rectangles inside
each container represent the volume of the job. The container is separated
into two halves by a dotted line. 
\par\end{raggedright}
\begin{raggedright}
The jobs are arranged in the bins according to $\prec$, with $J_{4}\prec J_{3}\prec J_{2}\prec J_{1}$
and $J_{6}\prec J_{5}$. $\beta(J,t)$ is the height of $J$ in this
figure - for example, $\beta(J_{2},t)=w(J_{3})+w(J_{4})$.
\par\end{raggedright}
\begin{raggedright}
The volume under $x$ is illustrated by a darker gray. For example:
\par\end{raggedright}
\begin{itemize}
\item \begin{raggedright}
For $J_{3}$ we have $\beta(J_{3},t)+w(J_{3})=w(J_{4})+w(J_{3})\le x$
and therefore $\gamma_{J_{3}}(x,t)=p_{t}(J_{3})$. 
\par\end{raggedright}
\item \begin{raggedright}
For $J_{2}$ we have that $\beta(J_{2},t)+\frac{w(J_{2})}{2}\le x<\beta(J_{2},t)+w(J_{2})$
giving that $\gamma_{J_{2}}(x,t)$ is the area of $J_{2}$ under the
dotted line. In other words, it is $\min(2^{i},p_{t}(J_{2}))=2^{i}$. 
\par\end{raggedright}
\item \begin{raggedright}
For $J_{5}$ we have that $\beta(J_{5},t)+\frac{w(J_{5})}{2}>x$,
which yields $\gamma_{J_{5}}(x,t)=0$.
\par\end{raggedright}
\end{itemize}
\raggedright{}As for processing, note that the algorithm will process
bin $A_{i+1}$ and not $A_{i}$, though their weights are equal. This
is since $J_{1}$ adds only $\frac{w(J_{1})}{2}$ to $\W_{A_{i}}(t)$,
because $J_{1}$'s volume is under its dotted line. Since the volume
of $J_{5}$ is above $J_{5}$'s dotted line, it adds its full weight,
$w(J_{5})$, to $\W_{A_{i+1}}(t)$.
\end{figure}
To prove Theorem \ref{thm:PTCompetitive}, we first prove Lemmas \ref{lem:NondecreasingArrival-2}
and \ref{lem:ImprovingArrival-1}. We then show that those lemmas,
together with Proposition \ref{prop:ExecutionMeansWholeVolume}, imply
the competitiveness of the algorithm.

Lemma \ref{lem:NondecreasingArrival-2} states that the volume of
a bin that is covered by a bar does not decrease upon the arrival
of a new job at that bin. Lemma \ref{lem:ImprovingArrival-1} states
that when a job arrives at a bin, raising the bar by three times the
weight of that new job is enough in order to gain its volume. A visualization
of lemma \ref{lem:NondecreasingArrival-2} is given in figure \ref{fig:Visualization-of-Lemma}.

\begin{obs}

If at some point in time $t$ we have that $J_{1}\prec_{t}J_{2}$,
then $J_{1}\prec_{t^{\prime}}J_{2}$ at any time $t^{\prime}$ in
which both $J_{1}$ and $J_{2}$ are alive. This is since processing
a job can only increase its priority, and the algorithm always processes
a job of maximum priority. We can therefore write $J_{1}\prec J_{2}$.

\end{obs}

\begin{figure}[!t]
\caption{\label{fig:Visualization-of-Lemma}Visualization of Lemma \ref{lem:NondecreasingArrival-2}}

\begin{centering}
\includegraphics[page=1]{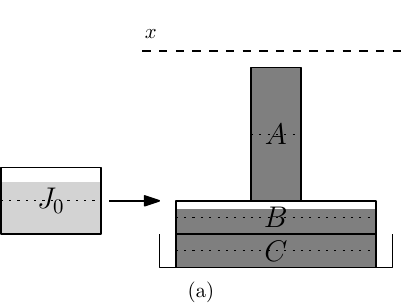}\includegraphics[page=2]{NondecreasingVolume}\\
\par\end{centering}
\raggedright{}The two images above, image (a) and image (b), show
a bin before and after the arrival of $J_{0}$. In both (a) and (b)
we have a bar at height $x$, with the volume covered by the bar colored
in dark gray. Lemma \ref{lem:NondecreasingArrival-2} states that
there is at least as much dark gray volume in image (b) as in image
(a).
\end{figure}
\begin{lem}
\label{lem:NondecreasingArrival-2}If a job $J_{0}$ is released at
time $t$ and assigned to $A_{i}$, and defining $t^{-}$ to be the
time $t$ prior to the event of $J_{0}$'s release, then for every
non-negative $x$ we have that 
\[
B_{A_{i}}(x,t)\ge B_{A_{i}}(x,t^{-})
\]
\end{lem}

\begin{proof}
Denote by $L\subseteq A_{i}(t)$ the set of jobs $J^{\prime}$ such
that $\gamma_{J^{\prime}}(x,t^{-})>\gamma_{J^{\prime}}(x,t)$. Note
that for every $J^{\prime}\in L$ we have $J_{0}\prec J^{\prime}$,
and thus $w(J^{\prime})\ge w(J_{0})$. 

For every $J^{\prime}\in L$, we define $\Delta_{J^{\prime}}=\gamma_{J^{\prime}}(x,t)-\gamma_{J^{\prime}}(x,t^{-})$
(note that from the definition of $L$, we have that $\Delta_{J^{\prime}}<0$).
It is enough to show 
\begin{equation}
\gamma_{J_{0}}(x,t)+\sum_{J^{\prime}\in L}\Delta_{J^{\prime}}\ge0\label{eq:ToShow_2a_P}
\end{equation}

Define $J_{\max}$ the maximal job in $L$ at time $t$ with respect
to $\prec_{t}$. We seperate into two cases.

\textbf{Case 1:} If $w(J_{\max})>w(J_{0})$, then due to the weights
being powers of $2$ we must have $w(J_{\max})\ge2w(J_{0})$. In this
case, the arrival of $J_{0}$ either:
\begin{itemize}
\item Changed $J_{\max}$ from being completely covered by $x$ at time
$t$, to having $x$ at least half the height of $J_{\max}$, but
below the top of $J_{\max}$.
\item Changed $J_{\max}$ from having $x$ at least half the height of $J_{\max}$,
to being below half the height of $J_{\max}$ but at least $J_{\max}$'s
base.
\end{itemize}
In either case, $J_{\max}$ has at most $2^{i}$ less volume covered
by $x$ at time $t$ than at $t^{-}$. In addition, since $x$ intersects
$J_{\max}$ at time $t$, all volume of jobs below $J_{\max}$ is
covered by $J_{\max}$. This yields $L=\{J_{\max}\}$, as well as
$\gamma_{J_{0}}(x,t)=p(J_{0})$. This completes case $1$.

\textbf{Case 2:}

Assume that $w(J_{max})=w(J_{0})$. Since $J_{0}\prec J_{\max}$,
we have that $p_{t}(J_{\max})\le p_{t}(J_{0})=p(J_{0})$. We separate
into the following subcases:
\begin{itemize}
\item Suppose $x$ is at least $J_{\max}$'s top at time $t^{-}$. It must
be at least $J_{\max}$'s bottom, and below $J_{\max}$'s top, at
time $t$. Since $x$ intersects $J_{\max}$ at time $t$, we have
$L=\{J_{\max}\}$ and $\gamma_{J_{0}}(x,t)=p(J)$. This gives us that
$\Delta_{J}\ge-p_{t}(J_{\max})\ge p(J)$, completing this subcase.
\item Suppose $x$ is at least $\beta(J_{\max},t^{-})+\frac{w(J_{\max})}{2}$
and below $J_{\max}$'s top at time $t^{-}$. Then at time $t$, $\beta(J_{\max},t)-\frac{w(J_{0})}{2}\le x<\beta(J_{\max},t)$.
Denoting by $J^{\prime}\in A_{i}(t)$ the job directly below $J_{\max}$
at time $t$, we have that $x$ intersects $J^{\prime}$ at time $t$. 
\begin{itemize}
\item If $J^{\prime}=J_{0}$, then $L=\{J_{\max}\}$. $J_{\max}$ lost its
lower half, which is at most $2^{i}$ volume. However, $J_{0}$'s
lower half is covered, and $p(J_{0})\ge2^{i}$. Therefore, $\gamma_{J_{0}}(x,t)=2^{i}$
as required. 
\item If $J^{\prime}\neq J_{0}$, then $L\subseteq\{J_{\max},J^{\prime}\}$.
$x$ lost $J_{\max}$'s lower half, at most $2^{i}$ volume. Since
$p(J_{0})\ge2^{i}$, this is at most the volume of $J_{0}$'s lower
half. It also lost $J^{\prime}$'s upper half, which has less volume
than $J_{0}$'s upper half (recall that $p(J_{0})\ge p_{t}(J^{\prime})$).
However, since $x$ intersects $J^{\prime}$ at time $t$, $\gamma_{J_{0}}(x,t)=p(J_{0})$,
as required.
\end{itemize}
\end{itemize}
This completes case 2, and the lemma.
\end{proof}
\begin{prop}
\label{prop:ShortJobs-2}At every time $t$ and bin $A_{i}$, there
exists at most one well-processed job $J\in A_{i}(t)$ of each weight.
\end{prop}

\begin{proof}
Since no job arrives well-processed, every job must become well-processed
through processing. For any possible weight $2^{j}$, once a job $J$
of that weight becomes well-processed the algorithm does not process
any non-well-processed job of the same weight in $A_{i}$ until $J$
is complete. This implies the lemma.
\end{proof}

\begin{cor}
\label{cor:ShortJobs-2}At every time $t$, bin $A_{i}$ and $j\in\Z$,
let $S\subseteq A_{i}(t)$ be the subset of well-processed jobs of
weight at most $2^{j}$. Then $w(S)<2^{j+1}$.
\end{cor}

\begin{proof}
Let $j_{\min}$ the minimal index such that $S$ has a job of weight
$2^{j_{\min}}$. Using Proposition \ref{prop:ShortJobs-2}:
\[
w(S)\le\sum_{k=j_{\min}}^{j}2^{k}<2^{j+1}
\]
\end{proof}

\begin{lem}
\label{lem:ImprovingArrival-1}If a job $J_{0}$ is released at time
$t$ and assigned to $A_{i}$, and defining $t^{-}$ to be the time
$t$ prior to the event of $J_{0}$'s release, then for every non-negative
$x$ we have that 
\[
B_{A_{i}}(x+3\cdot w(J_{0}),t)\ge B_{A_{i}}(x,t^{-})+p(J_{0})
\]
\end{lem}

\begin{proof}
Note that for every job $J^{\prime}\in A_{i}(t^{-})$, the base of
$J^{\prime}$ can increase by at most $w(J_{0})$ upon the arrival
of $J_{0}$. Thus, since the bar $x$ is also raised by $3\cdot w(J_{0})$
(more than $w(J_{0})$), we have $\gamma_{J^{\prime}}(x+3\cdot w(J_{0}),t)\ge\gamma_{J^{\prime}}(x,t^{-})$.
Therefore, it is enough to find a subset $S\subseteq A_{i}(t)$ such
that the volume of $S$ covered by $x+3\cdot w(J_{0})$ at $t$ is
at least $p(J_{0})$ more than the volume of $S$ covered by $x$
at $t^{-}$.

We observe the position of $J_{0}$ at time $t$. If $J_{0}$ is covered
by $x+3w(J_{0})$, that is $\beta(J_{0},t)\le x+2w(J_{0})$, then
$\gamma_{J_{0}}(x+3w(J_{0}),t)=p(J_{0})$. Noting that $J_{0}$ did
not exist at $t^{-}$, we choose $S=\{J_{0}\}$ and the proof is complete.

Otherwise, assume $\beta(J_{0},t)>x+2w(J_{0})$. We consider the set
of jobs that begin and end in the interval $[x,x+3w(J_{0})]$ \textendash{}
that is, jobs $J^{\prime}$ such that $\beta(J^{\prime},t)\ge x$
and $\beta(J^{\prime},t)+w(J^{\prime})\le x+3w(J_{0})$. Each job
$J^{\prime}\in S$ has the property that $\gamma_{J^{\prime}}(x,t^{-})=0$
and $\gamma_{J^{\prime}}(x+3w(J_{0}),t)=p_{t}(J^{\prime})$. To complete
the proof, it only remains to be seen that $p_{t}(S)\ge p(J_{0})$.

Observe that the interval $[x,x+3w(J_{0})]$ can only contain jobs
from the set $S$, and possibly some part of the job immediately below
$S$ (denoted $J_{\bot}$) and the job immediately above $S$ (denoted
$J_{\top}$). Therefore, we must have that $w(S)\ge3w(J_{0})-w(J_{\bot})-w(J_{\top})$. 

Since $S$ lies wholly below $J_{0}$, we must have that $J_{\bot},J_{\top}$
have weight at most $w(J_{0})$ \textendash{} thus, $S$ is non-empty.
We observe two cases, based on the maximal weight of a job in $S$.

\textbf{Case 1: $\max_{J\in S}w(J)=w(J_{0})$}

$S$ contains a job $J^{\prime}$ of weight $w(J_{0})$, but $J^{\prime}$
is below $J_{0}$, yielding $p_{t}(J^{\prime})\ge p_{t}(J_{0})=p(J_{0})$
and completing the case.

\textbf{Case 2: $\max_{J\in S}w(J)\le\frac{w(J_{0})}{2}$}

In this case, we show that $S$ contains two jobs that are not well-processed,
yielding $p_{t}(S)\ge2^{i+1}\ge p(J_{0})$ and completing the proof.
$J_{\bot}$ is below $S$, and thus $w(J_{\bot})\le\frac{w(J_{0})}{2}$,
yielding $w(S)\ge\frac{3w(J_{0})}{2}$. However, due to Lemma \ref{cor:ShortJobs-2},
the total weight of well-processed jobs of weight at most $\frac{w(J_{0})}{2}$
is less than $\frac{w(J_{0})}{2}$. Thus, the total weight of the
jobs of $S$ that are not well-processed is more than $\frac{w(J_{0})}{2}$,
which means at least two jobs. 

\end{proof}
We now show that Lemmas \ref{lem:NondecreasingArrival-2} and \ref{lem:ImprovingArrival-1}
imply the competitiveness of the algorithm. The outline of the proof
is shown in figure \ref{fig:Visualization-of-Competitiveness}, which
shows the state of the algorithm (on the left) and the optimum (on
the right) at any time $t$. First, Lemma \ref{lem:TrackingLemma-2}
shows that the bar $3\cdot W^{\ast}(t)$ covers the entire volume
in the algorithm. Lemma \ref{lem:GoodBinsImplyCompetitive} then uses
that fact to show that the bar $6\cdot W^{\ast}(t)$ covers the entire
\emph{weight }in the algorithm (in the figure, it covers the top of
each bin). The fact that there are $O(\log P)$ bins then implies
competitiveness.
\begin{figure}[!t]
\caption{\label{fig:Visualization-of-Competitiveness}Visualization of Competitiveness
Argument}

\centering{}\includegraphics[page=1]{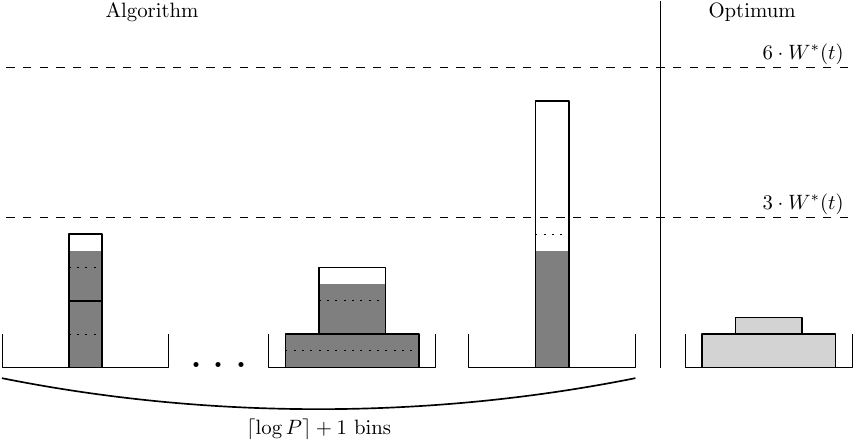}\\
\end{figure}

The following proposition makes use of the choice of bin made by the
algorithm.
\begin{prop}
\label{prop:ExecutionMeansWholeVolume}At a time $t$ and a non-negative
$x$, if $B(x,t)$ decreases as a result of the algorithm's processing,
then $B(x,t)=V(t)$.
\end{prop}

\begin{proof}
Note that $\W_{A_{i}}(t)$ as defined before is exactly the minimal
bar that covers all volume in $A_{i}$ (or equivalently, all volume
of the top job in $A_{i}$). Also note that a bar can only lose volume
from the processing of a job if it covers the entire volume of that
job. 

Let $x$ be such that $B(x,t)$ decreases through processing. Since
the job being processed is the top job of some bin $A_{i}$, this
implies $x\ge\W_{A_{i}}(t)$. From the algorithm's definition, $\W_{A_{i}}(t)\ge\W_{A_{i^{\prime}}}(t)$
for all $A_{i^{\prime}}\in\A$, yielding that $x$ covers all volume
in $A_{i^{\prime}}$. Summing, we have $B(x,t)=V(t)$.
\end{proof}
\begin{defn}
For any two points in time $t,t^{\prime}$ such that $t^{\prime}\in[0,t]$,
let $Q_{t}(t^{\prime})$ be the subset of $Q(t)$ that was alive at
time $t^{\prime}$. Formally, $Q_{t}(t^{\prime})=Q(t)\cap Q(t^{\prime})$.
\end{defn}

\begin{lem}
\label{lem:TrackingLemma-2} At every time $t$, a bar at three times
the weight of the optimum covers at least the volume of the optimum.
Formally:
\[
V^{\ast}(t)\le B(3\cdot w(Q^{\ast}(t)),t)
\]
\end{lem}

\begin{proof}
To show the lemma, we show that for every $t^{\prime}\in[0,t]$ we
have: 
\begin{equation}
p_{t^{\prime}}^{*}(Q_{t}^{*}(t^{\prime}))\le B(3\cdot w(Q_{t}^{*}(t^{\prime})),t^{\prime})\tag*{(*)}\label{eq:Recursive}
\end{equation}
When $t^{\prime}=t$, Inequality \ref{eq:Recursive} yields the lemma.
We show \ref{eq:Recursive} by induction, as $t^{\prime}$ advances
from $0$ to $t$. 

We can see that for $t^{\prime}=0$, before the release of any jobs,
\ref{eq:Recursive} is trivially true.

Now consider all $k$ jobs that arrive in the interval $[0,t]$, and
denote their release times by $t_{1},...,t_{k}$, so that in time
$t_{i}$ the $i$'th job has already been released. For a job $i$,
let $t_{i}^{-}$ be the time $t_{i}$ before the event of the release
of the $i$'th job. For convenience, denote $t_{0}=0$ and $t_{k+1}^{-}=t$.
We show the following claims:

\textbf{Claim 1}: If \ref{eq:Recursive} holds for $t^{\prime}=t_{i}$
such that $0\le i\le k$, then \ref{eq:Recursive} holds for every
$t^{\prime}\in[t_{i},t_{i+1}^{-}]$.

Since no job arrives in that interval, we have that $p_{t^{\prime}}^{*}(Q_{t}^{*}(t^{\prime}))$
is non-increasing as $t^{\prime}$ increases. It remains to consider
all cases in which $B(3\cdot w(Q_{t}^{*}(t^{\prime})),t^{\prime})$
decreases as a result of the algorithm's processing. Let $t^{\prime}$
be such that $B(3\cdot w(Q_{t}^{*}(t^{\prime})),t^{\prime})$ decreases
as a result of execution. Then using Proposition \ref{prop:ExecutionMeansWholeVolume},
we have that $B(3\cdot w(Q_{t}^{*}(t^{\prime})),t^{\prime})=V(t^{\prime})=V^{*}(t^{\prime})\ge p_{t^{\prime}}^{*}(Q_{t}^{*}(t^{\prime}))$
as required. 

\textbf{Claim 2:} If \ref{eq:Recursive} holds for $t^{\prime}=t_{i}^{-}$
such that $1\le i\le k$, then \ref{eq:Recursive} holds for $t^{\prime}=t_{i}$.

We'll observe the following cases.

\textbf{Case 1: }The $i$'th job, denoted $J$, is such that $J\notin Q_{t}^{*}(t_{i})$.
we have that $Q_{t}^{*}(t_{i})=Q_{t}^{*}(t_{i}^{-})$, and thus
\[
B(3\cdot w(Q_{t}^{*}(t_{i})),t_{i}^{-})=B(3\cdot w(Q_{t}^{*}(t_{i}^{-})),t_{i}^{-})\ge p_{t_{i}^{-}}^{*}(Q_{t}^{*}(t_{i}^{-}))=p_{t_{i}}^{*}(Q_{t}^{*}(t_{i}))
\]

Denoting by $A\in\A$ the bin to which $J$ has been assigned, and
using Lemma \ref{lem:NondecreasingArrival-2}, we have that $B_{A}(3\cdot w(Q_{t}^{*}(t_{i})),t_{i})\ge B_{A}(3\cdot w(Q_{t}^{*}(t_{i})),t_{i}^{-})$.
As for any other bin $A^{\prime}\in\A$, since $A^{\prime}$ has not
changed upon the release of $J$ we must have that $B_{A^{\prime}}(3\cdot w(Q_{t}^{*}(t_{i})),t_{i}^{-})=B_{A^{\prime}}(3\cdot w(Q_{t}^{*}(t_{i})),t_{i})$.
This gives us that $B(3\cdot w(Q_{t}^{*}(t_{i})),t_{i})\ge B(3\cdot w(Q_{t}^{*}(t_{i})),t_{i}^{-})$,
and thus the claim for this case.

\textbf{Case 2: }The $i$'th job, denoted $J$, is such that $J\in Q_{t}^{*}(t_{i})$.
we have that $Q_{t}^{*}(t_{i})=Q_{t}^{*}(t_{i}^{-})\cup\{J\}$, and
thus
\[
B(3\cdot w(Q_{t}^{*}(t_{i}^{-})),t_{i}^{-})\ge p_{t_{i}^{-}}^{*}(Q_{t}^{*}(t_{i}^{-}))=p_{t_{i}}^{*}(Q_{t}^{*}(t_{i}))-p(J)
\]

Denoting by $A\in\A$ the bin to which $J$ has been assigned, using
Lemma \ref{lem:ImprovingArrival-1}, we have that $B_{A}(3\cdot w(Q_{t}^{*}(t_{i})),t_{i})=B_{A}(3\cdot w(Q_{t}^{*}(t_{i}^{-}))+3\cdot w(J),t_{i})\ge B_{A}(3\cdot w(Q_{t}(t_{i}^{-})),t_{i}^{-})+p(J)$.
As for any other bin $A^{\prime}\in\A$, as in the previous case we
must have that $B_{A^{\prime}}(3\cdot w(Q_{t}^{*}(t_{i}^{-})),t_{i}^{-})=B_{A^{\prime}}(3\cdot w(Q_{t}^{*}(t_{i}^{-})),t_{i})\le B_{A^{\prime}}(3\cdot w(Q_{t}^{*}(t_{i})),t_{i})$.
Summing over the bins gives us the claim for this case.

Claims 1 and 2 show \ref{eq:Recursive} for every $t^{\prime}\in[0,t]$,
completing the proof.
\end{proof}

\begin{lem}
\label{lem:GoodBinsImplyCompetitive}Denoting by $\A^{\prime}\subset\A$
the set of non-empty bins at a time $t$, we have: 
\[
W(t)\le6\cdot|\A^{\prime}|\cdot W^{*}(t)
\]
\end{lem}

\begin{proof}
For every time $t$, using Lemma \ref{lem:TrackingLemma-2}, we have
that $B(3\cdot w(Q^{*}(t)),t)\ge V^{*}(t)$

Now, we can assume without loss of generality that that the optimal
algorithm is never idle when a job remains uncompleted. This implies
$V(t)=V^{\ast}(t)$, and thus $B(3\cdot w(Q^{*}(t)),t)=V(t)$. 

Denoting $x=3\cdot w(Q^{*}(t))$, for every bin $A\in\A$, observe
$J=\ttop_{A}(t)$. Since $x$ covers all of $J$'s volume, and thus
$x\ge\beta(J,t)+\frac{w(J)}{2}=w(A(t)\backslash\{J\})+\frac{w(J)}{2}$.
Thus, we have that $2x\ge2w(A(t)\backslash\{J\})+w(J)\ge w(A(t))$.
Therefore, $w(A(t))\le6\cdot w(Q^{*}(t))$, and summing over all $\A^{\prime},$
$W(t)\le6\cdot|\A^{\prime}|\cdot W^{*}(t)$.
\end{proof}
We can now prove the main theorem.
\begin{proof}
(of Theorem \ref{thm:PTCompetitive}) The algorithm assigns jobs to
at most $\lceil\log P\rceil+1$ bins. Using Lemma \ref{lem:GoodBinsImplyCompetitive},
we therefore have that: 
\[
W(t)\le6(\lceil\log P\rceil+1)\cdot W^{*}(t)
\]

This gives us that the algorithm is $O(\log P)$-competitive.
\end{proof}

\section{\label{sec:DensityAlgo}The $O(\log D)$-Competitive Algorithm}

We now describe an $O(\log D)$-competitive algorithm. As in Section
\ref{sec:ProcessingTimeAlgo}, this algorithm has an assignment part
and a processing part. The bins themselves are the set $\A=\{A_{i}\,|\,i\in\Z\}$. 

This algorithm makes the assumption that every job $J$ arrives so
that for some integer $i$ we have $d(J)=2^{i}$. This assumption
can be enforced by rounding the weight of each incoming job up by
a factor of at most $2$ to give the desired density. This only adds
a factor of $2$ to the competitive ratio of the algorithm, similarly
to Section \ref{sec:ProcessingTimeAlgo}.
\begin{defn}
For some positive $x$, define $\lg x=\lfloor\log_{2}x\rfloor$. 
\end{defn}

We now redefine the ordering $\prec_{t}$ within a bin.

\begin{defn}
For every bin $A$ at time $t$, and for distinct $J_{1},J_{2}\in A(t)$,
we write $J_{1}\prec_{t}J_{2}$ if $\lg(w(J_{1}))<\lg(w(J_{2}))$,
or if $\lg(w(J_{1}))=\lg(w(J_{2}))$ and $J_{2}$ is partially processed.
If $\lg(w(J_{1}))=\lg(w(J_{2}))$ and both $J_{1},J_{2}$ are not
partially processed, we break ties according to the indices of the
jobs, $I(J_{1})$ and $I(J_{2})$ (i.e. arbitrarily).
\end{defn}

Note that the previous definition did not address the case of $\lg(w(J_{1}))=\lg(w(J_{2}))$
and $J_{1}$, $J_{2}$ are both partially processed. Proposition \ref{prop:ShortJobsDensity}
shows that this case never happens, and thus the ordering is well
defined.

We recall the definition $\ttop_{A}(t)$ as the maximal job in bin
$A$ at time $t$ with respect to $\prec_{t}$.
\begin{defn}
For a bin $A=A_{i}\in\A$ and a time $t$, define 
\[
\W_{A}^{\prime}(t)=w(A(t)\backslash\{\ttop_{A}(t)\})+2^{\lg(w(\ttop_{A}(t)))}+\frac{p_{t}(\ttop_{A}(t))}{2^{i}}
\]
\end{defn}

\SetAlgoNoEnd

Algorithm \ref{alg:DAlgo} as described below is $O(\log D)$ competitive.

\begin{algorithm}[h]
\caption{\label{alg:DAlgo}$O(\log D)$ Competitive}

Whenever a new job $J$ arrives:\\
\Indp
assign $J$ to $A_{\log_{2}(d(J))}$\\
\Indm
At any time $t$:\\
\Indp
For $A=\arg \max _A(\mathscr{W}^{\prime}_A(t))$, process $\text{top}_A(t)$.\\
\Indm
\end{algorithm}

\section{\label{sec:DensityAnalysis}Analysis of $O(\log D)$-Competitive
Algorithm}

In this section, we prove the following theorem.
\begin{thm}
\label{thm:DCompetitive}Algorithm \ref{alg:DAlgo} is $O(\log D)$-competitive.
\end{thm}

For the purpose of viewing the volume covered by a bar, we add a dummy
job of each weight class to each bin. That is, a dummy job of weight
$2^{i}$ is added for any integer $i$. Each dummy job has higher
priority than (i.e. lies above) any other job of its weight class.
Note that for each real (non-dummy) job $J$, the total weight of
dummy jobs below $J$ sums to $2^{\lg(w(J))}$. The dummy jobs have
no volume, are never processed and are only for the purpose of modifying
the base of real jobs.

As in the processing time algorithm, we define the base of a job $J$
at time $t$ to be $\beta(J,t)=w(\{J^{\prime}\in A_{i}(t)\,|\,J^{\prime}\prec_{t}J\})$.

We also redefine the volume under a bar. Informally, a bar $x$ covers
exactly the volume that appears underneath $x$ in the visualization.
Formally:
\begin{defn}
The volume of job $J$ under bar $x$ is: 
\[
\gamma_{J}(x,t)=\begin{cases}
p_{t}(J) & x\ge\beta(J,t)+\frac{p_{t}(J)}{2^{i}}\\
2^{i}(x-\beta(J,t)) & \beta(J,t)\le x<\beta(J,t)+\frac{p_{t}(J)}{2^{i}}\\
0 & \text{otherwise}
\end{cases}
\]
\end{defn}

\begin{obs}\label{obs:HeightIsMinimalCoverDensity}

$\W_{A}^{\prime}(t)$ as defined before is the minimal bar that covers
the entire volume of a bin $A$.

 \end{obs}

\begin{figure}[!t]
\caption{\label{fig:Density-Bins}Density Bins}

\begin{raggedright}
The following image is a possible state of a bin in the algorithm
at some time $t$.
\par\end{raggedright}
\begin{centering}
\includegraphics{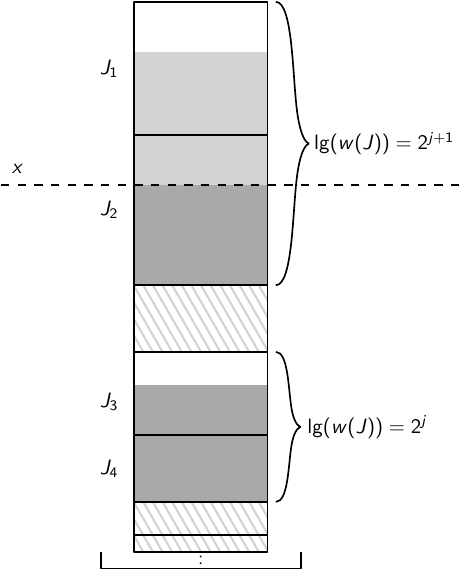}\\
\par\end{centering}
\begin{raggedright}
In this figure, the dummy jobs are dashed with gray lines. In the
base of the bin is the infinite base of dummy jobs, the height of
which sums to $2^{j}$. The $4$ real jobs, $J_{1}$ through $J_{4}$,
are stacked in the bin according to $\prec$, such that $J_{4}\prec J_{3}\prec J_{2}\prec J_{1}$.
Each real job is again represented as a container, in which the gray
area is the remaining processing time of the job. Its height is the
weight of the job, and its width is $2^{i}$ for the bin $A_{i}$.
Since in $A_{i}$ all jobs have density $2^{i}$, the area of the
container is exactly the initial volume of the job - that is, all
containers arrive full.
\par\end{raggedright}
\raggedright{}A bar $x$ is shown in the figure. The dark gray area
in the figure is the volume covered by the bar $x$, which changes
continuously with $x$ (compare this with covered volume in Section
\ref{sec:AnalysisPT}, which has discrete ``thresholds'').
\end{figure}

Consider any bin $A_{i}$. Every job $J$ assigned to $A_{i}$ has
$d(J)=2^{i}$. We prove the following lemma.
\begin{defn}
Define the following:
\end{defn}

\begin{itemize}
\item The length of an interval $[a,b]$, denoted by $\tau([a,b])=\begin{cases}
b-a & b\ge a\\
0 & \text{otherwise}
\end{cases}$.
\item For any job $J$ and any time $t$, define $\lambda(J,t)$ to be the
interval of height in which $J$ ``has volume''. Formally:
\[
\lambda(J,t)=[\beta(J,t),\beta(J,t)+\frac{p_{t}(J)}{2^{i}}]
\]
\end{itemize}
\begin{obs}

\label{obs:Interval-1}For any point in time $t$, for any job $J\in A_{i}(t)$,
and for any non-negative number $x$, 
\[
\gamma_{J}(x,t)=2^{i}\cdot\tau([0,x]\cap\lambda(J,t))
\]

\end{obs}

\begin{lem}
\label{lem:DensityNondecreasingArrival}If a job $J_{0}$ is released
at time $t$ and assigned to a bin $A_{i}$, and defining $t^{-}$
to be the time $t$ prior to the event of $J_{0}$'s release, then
for every non-negative $x$:
\[
B_{A}(x,t)\ge B_{A}(x,t^{-})
\]
\end{lem}

\begin{proof}
For jobs $J^{\prime}\in A(t^{-})$ such that $J^{\prime}\prec J_{0}$,
$\beta(J^{\prime},t^{-})=\beta(J^{\prime},t)$, and thus $\gamma_{J^{\prime}}(x,t)=\gamma_{J^{\prime}}(x,t^{-})$.
Now, consider $S$ the set of jobs so that $\forall J^{\prime}\in S:\,J_{0}\prec J^{\prime}$.
We have that $\beta(J^{\prime},t)=\beta(J^{\prime},t^{-})+w(J_{0})$
for all $J^{\prime}\in S$. Observe the following cases:

\textbf{Case 1}: $x<\beta(J_{0},t)$.

For any job $J^{\prime}\in S$, we have that 
\[
\beta(J^{\prime},t^{-})\ge\beta(J_{0},t)>x
\]

and thus $\gamma_{J^{\prime}}(x,t^{-})=0$, hence $x$ cannot cover
less volume of $J^{\prime}$ at time $t$. This completes the proof
for this case.

\textbf{Case 2}: $x\ge\beta(J_{0},t)$.

Considering any job $J^{\prime}\in S$, we have that $\beta(J^{\prime},t^{-})+w(J_{0})=\beta(J^{\prime},t)$.
Using Observation \ref{obs:Interval-1} we express the volume of $J^{\prime}$
covered at time $t^{-}$ using the state at time $t$:
\[
\gamma_{J^{\prime}}(x,t^{-})=\gamma_{J^{\prime}}(x+w(J_{0}),t)=2^{i}\cdot\tau([0,x+w(J_{0})]\cap\lambda(J^{\prime},t))
\]

This allows us to express $\Delta_{J^{\prime}}$, the loss of volume
under $x$ from $t^{-}$ to $t$:
\begin{multline*}
\Delta_{J^{\prime}}=\gamma_{J^{\prime}}(x,t^{-})-\gamma_{J^{\prime}}(x,t)=2^{i}\cdot\tau([0,x+w(J_{0})]\cap\lambda(J^{\prime},t))-\\
2^{i}\cdot\tau([0,x]\cap\lambda(J^{\prime},t))=\\
2^{i}\cdot\tau([x,x+w(J_{0})]\cap\lambda(J^{\prime},t))
\end{multline*}

The $\Delta_{J^{\prime}}$ for each $J^{\prime}\in S$ get disjoint
``portions'' of the interval $[x,x+w(J_{0})]$, since the $\lambda$-intervals
of different jobs do not overlap. They also can't get any portion
of the interval $[x,x+w(J_{0})]\cap\lambda(J_{0},t)$, which is occupied
by $J_{0}$. Therefore:

Therefore, we have:
\begin{align}
\sum_{J^{\prime}\in S}(\gamma_{J^{\prime}}(x,t^{-})-\gamma_{J^{\prime}}(x,t)) & \le2^{i}\cdot(\tau([x,x+w(J_{0})])-\tau([x,x+w(J_{0})]\cap\lambda(J_{0},t)))\nonumber \\
 & =2^{i}\cdot w(J_{0})-2^{i}\cdot\tau([0,x+w(J_{0})]\cap\lambda(J_{0},t))+2^{i}\cdot\tau([0,x]\cap\lambda(J_{0},t))\nonumber \\
 & =p(J_{0})-\gamma_{J_{0}}(x+w(J_{0}),t)+\gamma_{J_{0}}(x,t)\nonumber \\
 & \underbrace{=}_{(*)}p(J_{0})-p(J_{0})+\gamma_{J_{0}}(x,t)=\gamma_{J_{0}}(x,t)\label{eq:DNotMuchLoss}
\end{align}

Equality $(*)$ is due to $x\ge\beta(J_{0},t)$. Using this, we have
that 
\begin{eqnarray*}
B_{A}(x,t) & = & B_{A}(x,t^{-})+\gamma_{J_{0}}(x,t)+\sum_{J^{\prime}\in A(t^{-})}(\gamma_{J^{\prime}}(x,t)-\gamma_{J^{\prime}}(x,t^{-}))\\
 & = & B_{A}(x,t^{-})+\gamma_{J_{0}}(x,t)+\sum_{J^{\prime}\in S}(\gamma_{J^{\prime}}(x,t)-\gamma_{J^{\prime}}(x,t^{-}))\\
 & \ge & B_{A}(x,t^{-})+\gamma_{J_{0}}(x,t)-\gamma_{J_{0}}(x,t)=B_{A}(x,t^{-})
\end{eqnarray*}

Where the inequality uses equation \ref{eq:DNotMuchLoss}. 
\end{proof}
\begin{prop}
\label{prop:ShortJobsDensity}At every time $t$, for every $i,j$,
there exists at most one partially processed job $J\in A(t)$ such
that $\lg(w(J))=j$.
\end{prop}

\begin{proof}
Once there is a single partially processed job $J$ in a specific
weight class, the algorithm will not work on another job of that class
until $J$ is complete. Thus there cannot be more than one partially
processed job in a weight class. 
\end{proof}

\begin{cor}
\label{cor:ShortJobsDensity}At any point in time $t$ and for $j\in\Z$,
let $S\subseteq A(t)$ be the set of partially-processed jobs $J$
such that $\lg(w(J))\le j$. Then $w(S)<4\cdot2^{j}$.
\end{cor}

\begin{proof}
For every $k\in\Z$ such that $k\le j$ we have through Proposition
\ref{prop:ShortJobsDensity} that there exists at most one partially
processed job $J_{k}\in A(t)$ such that $\lg(w(J_{k}))=k$. The weight
of $J_{k}$ is at most $2^{k+1}$. Defining $m=\min_{J\in S}\lg(w(J))$,
we have that 
\[
w(S)\le\sum_{k=m}^{j}2^{k+1}<2^{j+2}\le4\cdot2^{j}
\]
\end{proof}
\begin{lem}
\label{lem:DensityIncreasingArrival}If a job $J_{0}$ is released
at time $t$ and assigned to bin $A_{i}$, and defining $t^{-}$ to
be the time $t$ prior to the event of $J_{0}$'s release, then for
every non-negative $x$ we have 
\[
B_{A}(x+10\cdot w(J_{0}),t)\ge B_{A}(x,t^{-})+p(J_{0})
\]
\end{lem}

\begin{proof}
Note that for any job $J^{\prime}\in A(t^{-})$, the base of $J^{\prime}$
can rise by at most $w(J_{0})$ upon the arrival of $J_{0}$. Since
the bar $x$ is raised by $10\cdot w(J_{0})$ (more than $w(J_{0})$),
the new bar must cover at least as much volume of $J^{\prime}$ at
$t$ as the old bar did at $t^{-}$. To complete the proof it is thus
enough to find a subset $S\subseteq A(t)$ such that $x$ covers at
least $p(J_{0})$ more of $S$'s volume at $t$ than at $t^{-}$.

Now, observe the position of $J_{0}$ at time $t$. If $J_{0}$ is
fully covered by $x+10\cdot w(J_{0})$, then $x+10\cdot w(J_{0})$
covers $J_{0}$'s entire volume, and choosing $S=\{J_{0}\}$ completes
the proof.

Otherwise, $\beta(J_{0},t)+w(J_{0})>x+10\cdot w(J_{0})$. Choose $S$
to be the set of jobs that start and end within $[x,x+10\cdot w(J_{0})]$
\textendash{} that is, for every $J\in S$ we have $\beta(J,t)\ge x$
and $\beta(J,t)+w(J)\le x+10\cdot w(J_{0})$.

We now claim that $w(S)\ge6\cdot w(J_{0})$. This is since the job
immediately above $S$ and the job immediately below $S$ can occupy
some of the interval $[x,x+10\cdot w(J_{0})]$, but the rest is taken
by $S$. The jobs immediately below and above $S$ must each have
weight at most $2w(J_{0})$, otherwise they would be above $J_{0}$,
in contradiction (recall that $J_{0}$ is above $S$).

Denote by $S^{\prime}\subseteq S$ the subset of real jobs in $S$.
The weight of dummy jobs below $J_{0}$ sums to $2^{\lg(w(J_{0}))}\le w(J_{0})$,
thus $w(S^{\prime})\ge5w(J_{0})$. 

Denote by $S^{\prime\prime}\subseteq S^{\prime}$ the subset of unprocessed
jobs in $S^{\prime}$. Through corollary \ref{cor:ShortJobsDensity},
the total weight of partially processed jobs in $S$ is at most $4\cdot2^{\lg(w(J_{0}))}\le4w(J_{0})$,
and thus $w(S^{\prime\prime})\ge w(J_{0})$. 

All jobs in $S^{\prime\prime}$ are real and unprocessed, giving $S^{\prime\prime}$
a total volume of at least $2^{i}\cdot w(J_{0})=p(J_{0})$. Since
$J_{0}$ is above all jobs of $S^{\prime\prime}$, they did not move
upon the arrival of $J_{0}$. Since they all lie in the interval $[x,x+10\cdot w(J_{0})]$,
all of their volume is covered by $x+10\cdot w(J_{0})$ at time $t$,
and none of their volume is covered by $x$ at $t^{-}$. This completes
the proof.
\end{proof}
With Lemmas \ref{lem:DensityNondecreasingArrival} and \ref{lem:DensityIncreasingArrival}
in hand, we can now repeat the process used in Section \ref{sec:AnalysisPT}
to show competitiveness, with slight changes to the relevant constants. 
\begin{prop}
\label{prop:DExecutionMeansEntireVolume}(analogue of Proposition
\ref{prop:ExecutionMeansWholeVolume}) At a time $t$ and a non-negative
$x$, if $B(x,t)$ decreases as a result of the algorithm's processing,
then $B(x,t)=V(t)$.
\end{prop}

\begin{proof}
Observe that $\W_{A}^{\prime}(t)$ is exactly the minimal bar that
covers the entire volume of the bin $A$ at time $t$. The remainder
of the proof is identical to \ref{prop:ExecutionMeansWholeVolume}.
\end{proof}
\begin{lem}
\label{lem:DTrackingLemma}(analogue of Lemma \ref{lem:TrackingLemma-2})
At every time $t$, a bar at $10$ times the weight of the optimum
covers at least the volume of the optimum. Formally:
\[
V^{\ast}(t)\le B(10\cdot w(Q^{\ast}(t)),t)
\]
\end{lem}

\begin{proof}
The proof is nearly identical to that of \ref{lem:TrackingLemma-2},
using Proposition \ref{prop:DExecutionMeansEntireVolume} and Lemmas
\ref{lem:DensityNondecreasingArrival} and \ref{lem:DensityIncreasingArrival}.
The change of constant from $3$ to $10$ results from swapping Lemma
\ref{lem:DensityIncreasingArrival} for Lemma \ref{lem:ImprovingArrival-1}.
\end{proof}
\begin{lem}
\label{lem:DGoodBinsImplyCompetitive}(analogue of Lemma \ref{lem:GoodBinsImplyCompetitive})
Denoting by $\A^{\prime}\subset\A$ the set of non-empty bins at a
time $t$, we have: 
\[
W(t)\le20\cdot|\A^{\prime}|\cdot W^{*}(t)
\]
\end{lem}

\begin{proof}
Similar proof to Lemma \ref{lem:GoodBinsImplyCompetitive}, using
Lemma \ref{lem:DTrackingLemma}. It is important to note that in the
proof of Lemma \ref{lem:GoodBinsImplyCompetitive}, we used the fact
that at any time $t$, if a bar $x$ covers the entire volume of a
bin, then $2x$ is at least the entire weight of the jobs in that
bin. This holds in our case as well, when considering the total weight
of non-dummy jobs in the bin.

To observe this, consider the bin without dummy jobs, and denote by
$\beta^{\prime}(J,t)$ the new base of a real job $J$. Consider any
such $J$. For $J$ to have volume covered under $x$ in the original
bin, we must have that $x\ge\beta(J,t)$. Since $\beta(J,t)=\beta^{\prime}(J,t)+2^{\lg(w(J))}$,
we have that $x\ge\beta^{\prime}(J,t)+2^{\lg(w(J))}\ge\beta^{\prime}(J,t)+\frac{w(J)}{2}$.
The bar $2x$ therefore covers $J$ completely. 
\end{proof}
We can now prove Theorem \ref{thm:DCompetitive}.
\begin{proof}
(of Theorem \ref{thm:DCompetitive}) Since the algorithm assigns jobs
to at most $\lceil\log D\rceil+1$ bins, Lemma \ref{lem:DGoodBinsImplyCompetitive}
gives us: 
\[
W(t)\le20(\lceil\log D\rceil+1)\cdot W^{*}(t)
\]

Using Observation \ref{obs:LocalCompImpliesGlobal}, we have that
the algorithm is $O(\log D)$-competitive.
\end{proof}

\section{\label{sec:MinAlgo}The $O(\log(\min(W,P,D)))$-Competitive Algorithm}

In this section we describe an algorithm which is $O(\min(\log(W,P,D)))$-competitive
without knowing $W,P,D$ in advance. As in Sections \ref{sec:ProcessingTimeAlgo}
and \ref{sec:DensityAlgo}, the algorithm is composed of a bin-assignment
part and a bin-processing part.

The main idea in the algorithm is combining bins for processing time,
bins for density and bins for weight. In this algorithm, all the bins
start closed, and must be opened prior to being assigned any jobs
by the algorithm. The algorithm only opens bins as triplets with a
bin of each type; this property keeps the number of bins logarithmic
in the minimal of $W,P,D$.

Define the disjoint sets of bins $\A_{1}=\{A_{i}\,|\,i\in\Z\}$ (processing
time bins), $\A_{2}=\{A_{i}^{\prime}\,|\,i\in\Z\}$ (density bins)
and $\A_{3}=\{A_{i}^{\text{\ensuremath{\prime\prime}}}\,|\,i\in\Z\}$
(weight bins). Our set of bins is $\A=\A_{1}\cup\A_{2}\cup\A_{3}$. 

Inside the processing-time bins and density bins, the jobs are ordered
as in Sections \ref{sec:ProcessingTimeAlgo} and \ref{sec:DensityAlgo}
respectively. Inside weight bins, the jobs are ordered according to
remaining processing time (lower processing time is higher).

Since the algorithm uses various ways of rounding the weights of jobs,
we cannot make the assumption that the weights are rounded a-priori.
Therefore the rounding of the weights is a part of the algorithm.

Let $\W_{A}$ be defined for $A\in\A_{1}$ as in Section \ref{sec:ProcessingTimeAlgo},
and let $\W_{A}^{\text{\ensuremath{\prime}}}$ be defined for $A\in\A_{2}$
as in Section \ref{sec:DensityAlgo}.
\begin{defn}
For a bin $A\in\A$ and a time $t$, we define the score of a bin
$A$ to be:

\[
\tilde{\W}_{A}(t)=\begin{cases}
\W_{A}(t) & A\in\A_{1}\\
\W_{A}^{\prime}(t) & A\in\A_{2}\\
w(A(t)) & A\in\A_{3}
\end{cases}
\]

\end{defn}

As in previous sections, we denote by $\ttop_{A}(t)$ the maximal
job in $A$ at time $t$ with respect to the ordering in $A$.

Algorithm \ref{alg:BinAssignerMin} as described below is $O(\min(\log(W,P,D)))$
competitive.

\SetAlgoNoEnd

\SetAlgoNoLine

\begin{algorithm}
\caption{\label{alg:BinAssignerMin}$O(\log(\min(P,D,W)))$ Competitive}

When a new job $J$ arrives:\\
\Indp
\If {$A_{i}$ such that $2^i<p(J)\le 2^{i+1}$ is open} 
{
assign $J$ to $A_{i}$\\
round the weight of $J$ up to $2^{\lg(w(J))+1}$
}
\ElseIf {$A^{\prime}_{\lg (d(J))}$ is open} 
{
assign $J$ to $A^{\prime}_{\lg (d(J))}$\\
round the weight of $J$ up to give $J$ the new density $2^{\lg (d(J))}$
}
\ElseIf {$A^{\prime \prime}_{\lg(w(J))+1}$ is open} 
{
assign $J$ to $A^{\prime \prime}_{\lg(w(J))+1}$\\
round the weight of $J$ up to $2^{\lg(w(J))+1}$
}
\Else 
{
open $A_{i}$, $A^{\prime}_{\lg (d(J))}$ and $A^{\prime \prime}_{\lg(w(J))+1}$\\
assign $J$ to $A^{\prime \prime}_{\lg(w(J))+1}$\\
round the weight of $J$ up to $2^{\lg(w(J))+1}$
}
\Indm
At any time $t$:\\
\Indp
For $A=\arg \max _A(\tilde{\mathscr{W}}_A(t))$, process $\text{top}_A(t)$
\end{algorithm}

\section{\label{sec:MinAnalysis}Analysis of $O(\log(\min(W,P,D)))$-Competitive
Algorithm}

We want to prove the following theorem.
\begin{thm}
\label{thm:MinCompetitive}The algorithm described in Section \ref{sec:MinAlgo}
is $O(\log(\min(W,P,D)))$-competitive.
\end{thm}

Note that since the weight rounding done by the algorithm is by a
factor of at most $2$, we henceforth consider competitiveness compared
to the optimum for the rounded instance.

In order to do so, we define the volume under a bar $x$ differently
for each type of bin. For a processing-time bin, it is as defined
in Section \ref{sec:AnalysisPT}. For a density bin, it is as defined
in Section \ref{sec:DensityAnalysis}. For a weight bin, it is defined
as the total volume of jobs that $x$ covers completely (as defined
in the $O(\log W)$-competitive algorithm of \cite{DBLP:journals/talg/BansalD07})

Note that from the definition of volume under bar in a weight bin
$A$, we have that $w(A(t))$ is the minimal bar which covers the
entire volume of $A$. Also note that if $x$ covers the entire volume
in a weight bin, it also covers the entire weight of that bin.

The proof of the following lemma is given in \cite{DBLP:journals/talg/BansalD07}:
\begin{lem}
For a bar $x$ and a job $J_{0}$ that arrives at a weight bin $A_{i}^{\prime\prime}$
at time $t$, then:
\begin{enumerate}
\item $B_{A_{i}^{\prime\prime}}(x,t)\ge B_{A_{i}^{\prime\prime}}(x,t^{-})$
\item $B_{A_{i}^{\prime\prime}}(x+w(J),t)\ge B_{A_{i}^{\prime\prime}}(x,t^{-})+p(J_{0})$
\end{enumerate}
\end{lem}

We can now prove Theorem \ref{thm:MinCompetitive}.
\begin{proof}
(of Theorem \ref{thm:MinCompetitive}) 

The proof is again analoguous to that in Section \ref{sec:AnalysisPT}.
The score for each bin is exactly the minimal bar that covers all
volume in that bin, enabling us to repeat Proposition \ref{prop:ExecutionMeansWholeVolume}. 

Note that the algorithm assigns jobs to bins such that:
\begin{itemize}
\item Whenever a job $J$ is assigned to a bin $A_{i}\in\A_{1}$ we have
that $w(J)=2^{j}$ for some $j\in\Z$ and $2^{i}\le p(J)\le2^{i+1}$. 
\item Whenever a job $J$ is assigned to a bin $A_{i}^{\prime}\in\A_{2}$
we have that $d(J)=2^{i}$. 
\item Whenever a job $J$ is assigned to a bin $A_{i}^{\prime\prime}\in\A_{3}$,
we have that $w(J)=2^{i}$. 
\end{itemize}
Under these conditions, we know that upon an arrival of a job $J$
to any bin $A$ at a time $t$ we have:
\begin{itemize}
\item $B_{A}(x,t)\ge B_{A}(x,t^{-})$
\item $B_{A}(x+10w(J),t)\ge B_{A}(x,t^{-})+p(J)$
\end{itemize}
Using these last two bullets, as well as the analogue for Proposition
\ref{prop:ExecutionMeansWholeVolume}, we obtain an analogue for Lemma
\ref{lem:TrackingLemma-2} which states that $V^{\ast}(t)\le B(10\cdot w(Q^{\ast}(t)),t)$
at any time $t$.

We also know that if $x$ covers all volume in the algorithm, $2x$
covers the entire weight in the algorithm. We thus obtain an analogue
of Lemma \ref{lem:GoodBinsImplyCompetitive} that yields $W(t)\le20\cdot(\#\text{open bins})\cdot W^{\ast}(t)$
at any time $t$.

The algorithm can only open $3\cdot(\lceil\log(\min(W,P,D))\rceil+1)$
bins. This is since it only opens a triplet of bins upon the arrival
of a job which does not fit in any existing bin. For example, if $W=\min(W,P,D)$,
after opening $\lceil\log W\rceil+1$ triplets the entire weight range
is covered, and no more triplets will be opened. The same argument
holds for $P$ and $D$.

Therefore we have:
\[
W(t)\le2\cdot10\cdot3\cdot(\lceil\log(\min(W,P,D))\rceil+1)\cdot W^{*}(t)=60\cdot(\lceil\log(\min(W,P,D))\rceil+1)\cdot W^{*}(t)
\]

Using Observation \ref{obs:LocalCompImpliesGlobal}, we have that
the algorithm is $O(\log(\min(W,P,D)))$-competitive.
\end{proof}
\bibliographystyle{plain}
\bibliography{bibFile}

\end{document}